\documentclass{article}

\usepackage{arxiv}
\usepackage{amsthm,amsmath,amssymb}
\usepackage[dvipsnames]{xcolor}
\usepackage[numbers]{natbib}
\usepackage[utf8]{inputenc} 
\usepackage[T1]{fontenc}    % use 8-bit T1 fonts
\usepackage{hyperref}       % hyperlinks
\usepackage{url}            % simple URL typesetting
\usepackage{booktabs}       % professional-quality tables
\usepackage{amsfonts}       % blackboard math symbols
\usepackage{nicefrac}       % compact symbols for 1/2, etc.
\usepackage{microtype}      % microtypography
\usepackage{lipsum}
\usepackage[justification=centering,skip=7pt,labelfont=bf]{caption}

\usepackage{natbib}
\usepackage{balance}

\usepackage{comment}
\usepackage{graphicx}
\usepackage{listings}
\usepackage{multicol}
\usepackage{dsfont}

\AtBeginDocument{%
  \providecommand\BibTeX{{%
    \normalfont B\kern-0.5em{\scshape i\kern-0.25em b}\kern-0.8em\TeX}}}

\usepackage{float}
\floatstyle{plaintop}
\restylefloat{table}

\makeatletter
\newtheorem*{rep@theorem}{\rep@title}
\newcommand{\newreptheorem}[2]{%
\newenvironment{rep#1}[1]{%
 \def\rep@title{#2 \ref{##1}}%
 \begin{rep@theorem}}%
 {\end{rep@theorem}}}
\makeatother

\newtheorem{theorem}{Theorem}[section]
\newreptheorem{theorem}{Theorem}
\newreptheorem{lemma}{Lemma}
\newtheorem{corollary}[theorem]{Corollary}
\newtheorem{lemma}[theorem]{Lemma}

\newtheorem{definition}[theorem]{Definition}

\usepackage{listings}
\usepackage{etoolbox}
\newtoggle{InString}{}% Keep track of if we are within a string
\togglefalse{InString}% Assume not initally in string
\newcommand*{\ColorIfNotInString}[1]{\iftoggle{InString}{#1}{\color{blue}#1}}%
\newcommand*{\ProcessQuote}[1]{#1\iftoggle{InString}{\global\togglefalse{InString}}{\global\toggletrue{InString}}}%
\definecolor{code_indent}{HTML}{CCCCCC}

\newenvironment{figureAsListing}
    {
    \addtocounter{figure}{-1}
    \refstepcounter{lstlisting}
     
    \begin{figure}[!htbp]
        
        \centering
    }
    { 
        \end{figure} 
    }

\usepackage{multicol}

\def\O{\mathcal{O}}

\DeclareMathOperator{\E}{\mathbb{E}~}

\title{Provably-Efficient and Internally-Deterministic \\ Parallel Union-Find}

\author{
Alexander Fedorov \\
  IST Austria\\
  \texttt{alexander.fedorov@ist.ac.at} \\
  \And
  Diba Hashemi \\
  IST Austria\\
  \texttt{diba.hashemi@ist.ac.at}
  \And
  Giorgi Nadiradze \\
  IST Austria\\
  \texttt{giorgi.nadiradze@ist.ac.at}\\
  \And
    Dan Alistarh\\
  IST Austria\\
    \texttt{dan.alistarh@ist.ac.at}  
}

\begin{document}

\maketitle

\begin{abstract}
    Determining the degree of inherent parallelism in classical sequential algorithms and leveraging it for fast parallel execution is a key topic in parallel computing, 
    and detailed analyses are known for a wide range of classical algorithms. 
    In this paper, we perform the first such analysis for the fundamental \emph{Union-Find} problem, in which we are given a graph as a sequence of edges, and must maintain its connectivity structure under edge additions. We prove that classic sequential algorithms for this problem are well-parallelizable under reasonable assumptions, addressing a conjecture by [Blelloch, 2017]. More precisely, we show via a new potential argument that, under uniform random edge ordering, parallel union-find operations are unlikely to interfere: $T$ concurrent threads processing the graph in parallel will encounter memory contention $\O(T^2 \cdot \log |V| \cdot \log |E|)$ times in expectation, where $|E|$ and $|V|$ are the number of edges and nodes in the graph, respectively. 
    We leverage this result to design a new parallel Union-Find algorithm that is both \emph{internally deterministic}, i.e., its results are guaranteed to match those of a sequential execution, but also \emph{work-efficient and scalable}, as long as the number of threads $T$ is $\O(|E|^{\frac{1}{3} - \varepsilon})$, for an arbitrarily small constant $\varepsilon > 0$, which holds for most large real-world graphs. We  present lower bounds which show that our analysis is close to optimal, and experimental results suggesting that the performance cost of internal determinism is limited. 
\end{abstract}


\section{Introduction}

A popular approach to  efficient parallelization has been to leverage the \emph{inherent parallelism} present in many sequential algorithms, and graph problems have been shown to be a particularly fertile ground for this approach. 
A simple illustration is given by the greedy algorithm for Maximal Independent Set (MIS) on graphs in which we initially fix a random ordering of the graph nodes, after which we process nodes one-at-a-time, in this order, adding the node to the MIS if none of its earlier neighbors has been previously added to the MIS, and discarding it otherwise. When trying to parallelize such sequential algorithms, two key questions are 1) the likelihood of \emph{contention}, i.e. that two arbitrary nodes in the ordering have data dependencies,  and 2) the depth of the \emph{longest dependency chain} between ``dependent’’ nodes, for any given graph. Intuitively, the total number of nodes divided by the maximal dependency depth gives an estimate for how many nodes can be processed in parallel, on average. Assuming low dependency depth, e.g., sublinear in the number of nodes or edges, a second challenge is to design parallel algorithms which are able to leverage it for fast end-to-end runtimes, usually measured as the total number of thread memory accesses or \emph{work}, which includes managing any auxiliary data structures.

A rich line of work has investigated these questions for many standard problems, such as MIS~\cite{GreedyMISAreParallel, FischerNoever}, two-dimensional linear programming~\cite{ParallelismInRandomized}, random permutation, list, and tree contraction~\cite{ParallelRandomPermutation}, or Delaunay mesh triangulation~\cite{ParallelismInRandomized}. For example, the (expected) dependency depth for the above greedy MIS algorithm was first shown to be $\O(\log^2 |V|)$ for any input graph with $|V|$ vertices~\cite{GreedyMISAreParallel}, which was later reduced to a tight $\O(\log |V|)$ via a technical breakthrough by Fischer and Noever~\cite{FischerNoever}. These analyses can lead to optimal-work algorithms, and are complemented by frameworks which allow speedups to be realized in practice~\cite{GraphAlgorithmsCanBeFast}.

One classic problem for which the potential for parallelization is still not well-understood is Disjoint-Set Union / Union-Find~\cite{Tarjan1975}, in which we are given a sequence of graph edges, and must maintain connectivity structure under edge additions. Specifically, assuming a random edge ordering and an algorithm implementing one of the standard sequential linking strategies, e.g. linking by rank~\cite{Tarjan1975}, we wish to analyze the likelihood that processing two edges in parallel may lead to a potential data race in the Union-Find data structure, and whether work-efficient parallel algorithms with deterministic computations exist for this problem.
While Union-Find was conjectured to be provably highly-parallel by Blelloch~\cite{SomeSequentialAreParallel}, this question has so far remained open. 

\paragraph{Contribution} In this paper, we take a significant step towards resolving this problem for natural parallel variants of the Union-Find algorithm. We start with a simple and general ``edge collision’’ model, which captures the possible data races that may arise when executing such algorithms. We then present a new potential argument showing that, under random edge ordering, the number of edge collisions among $T$ parallel threads can be bounded by $\O(T^2  \log |V|  \log |E| )$ in expectation. This implies an  upper bound of expected $\O(|E|^{\frac 2 3} \cdot \log^{\frac 1 3} |E| \cdot \log^{\frac 1 3} |V|)$ on the dependency depth of this problem. Conversely, we show an iteration lower bound of $\Omega \left(\frac {\sqrt {|E|}} {\sqrt {\log {|E|}}}\right)$ for cycle graphs w.h.p. for any parallel algorithm which processes maximally-large edge batches without violating dependencies, which was suggested in~\cite{ParallelismInRandomized} as a general approach for randomized incremental algorithms in this setting.

Based on this new analysis, we propose a simple parallel algorithm with the total number of steps (work complexity) $\O( |E| \alpha( |V| ) + T^3  \alpha( |V| )  \log |V|  \log |E| )$ in expectation, on any graph with $|E|$ edges and $|V|$ nodes, using $T$ concurrent processors. (Here, $\alpha( \cdot )$ denotes the inverse-Ackermann function arising in the analysis of sequential Union-Find.) Thus, the overhead of parallelism, given by the second term, is negligible as long as $T = \O(|E|^{\frac{1}{3} - \varepsilon} )$, for any constant $\varepsilon > 0$, which should be reasonable in practice, as graph inputs tend to be large relative to parallelism. 

Our main technical contribution is in the potential argument bounding the collision probability between tasks processing two distinct edges in the random ordering. Specifically, we assume a sequential algorithm maintaining a standard ``compressed forest'' data structure, where elements are nodes, arranged into directed trees, and the tree root is the set representative~\cite{Tarjan1975}. The addition of a new edge may lead to  components being linked, where the link direction is decided by the algorithm. Two edges \emph{collide} if they cannot be processed in parallel: for instance, this can happen when the two edges processed simultaneously would lead to the same root being linked to  \emph{different} components. 

In this context, we first analyze a ``sequentialized'' variant of the process, in which, at each step, a new edge is added to the data structure, and consider the probability $p_t$ of two randomly chosen edges ``colliding'' after exactly $t$ steps. 
While we cannot bound $p_t$ independently of the graph structure, and $p_t$ may fluctuate significantly over steps, we are able to bound the \emph{average} value of $p_t$, taken over time steps $t$, via a new 
potential function, which we show to be well-correlated with the ``smoothed'' collision probabilities over time, $\sum\limits_{t = 0}^{|E|-1} p_t$. Specifically, we show that, for any linking strategy that bounds the maximum forest depth to $\O( \log | V |)$, the sum of collision probabilities over all edges will satisfy $\sum\limits_{t = 0}^{|E| - 1} p_t = \O( \log |V| \log |E|).$  
In turn, this implies a bound on the expected number of collisions over a number of parallel steps. 
We present a complementary lower bound showing that the number of collisions for any Union-Find algorithm is $\Omega( \log |V|)$ for cycle graphs. 

Next, we design a work-efficient algorithm which leverages these analytical observations. We apply the \emph{deterministic reservations} approach~\cite{InternallyDeterministicParallel}, which we customize to our setting. Specifically, for a well-chosen parameter $S$, the algorithm proceeds in ``windows'' of $S$ consecutive edges, where in each such stage the threads first attempt to ``mark'' roots in parallel via deterministic reservations, and then examine whether any of the reservations resulted in data conflicts because of the underlying graph structure. If no such conflict occurs, then the edges can be processed fully in parallel. Otherwise, the threads process the conflict-free prefix. Then, they proceed to execute a new window of size $S$ on the remaining suffix. The collision bound above implies that this algorithm has asymptotically-optimal work if the number of threads is $\O(|E|^{\frac 1 3 - \varepsilon})$. 

Our algorithm is \emph{internally deterministic}---roughly, given a fixed input ordering, one obtains a unique dependency graph between the tasks corresponding to the edges, and will therefore have the same complexity as an equivalent sequential execution. As a consequence of internal determinism, our algorithm should be generally-useful as an efficient sub-procedure: for example, we can leverage it for a new solution for the Dynamic Spanning Tree and the Minimal Spanning Tree (MST) problems. 
Our main result is as follows:

\begin{theorem}
There exists an internally-deterministic parallel Union-Find algorithm for CRCW PRAM model with \emph{priority writes} that has $\O(|E| \cdot \alpha(|V|))$ expected total work and $\O(|E|^{\frac 2 3} \cdot \emph{polylog}(|E|))$ parallel depth on a randomly shuffled sequence of edges and for the number of parallel threads $T = \O(|E|^{\frac 1 3 - \varepsilon})$. 
This also implies a parallel version of Kruskal's MST algorithm, which has the same work and depth bounds in the special case where we have a sorted sequence of edges, and edge weights are generated so that the sorted sequence is a random shuffle. 
\end{theorem}

\paragraph{Related Work}
Our work extends the line of research analyzing inherent parallelism in classical sequential algorithms~\cite{GreedyMISAreParallel, FischerNoever, ParallelismInRandomized, ParallelRandomPermutation, GraphAlgorithmsCanBeFast} to show that Union-Find is also efficiently parallelizable. This partly addresses a question by Blelloch~\cite{SomeSequentialAreParallel}, who posed the dependency depth of Union-Find as an  open problem.

Prior work proposed a linear-work algorithm for graph connectivity~\cite{ParallelConnectedComponents}; however, this uses a different decomposition-based parallelization approach, which requires a static graph, and does not allow for incremental edge additions, nor online connectivity queries. 
The best known parallel algorithm for the a variant of Union-Find in the concurrent-read concurrent-write (CRCW) PRAM model was proposed by Simsiri et al.~\cite{WorkEfficientUnionFind}. They investigate a parallel version of \emph{batched} union-find, assuming that operations are inherently grouped into batches. The algorithm is work-efficient and guarantees  polylogarithmic span. It processes batches of \texttt{find} operations with \emph{path compaction} and asymptotically the same total work as in the sequential setting; however, this part requires significant synchronization between threads, and thus, it may not be very efficient in practice. For \emph{union} operations, the approach is to reduce the problem to a linear-work parallel connected components algorithm, such as \cite{ParallelConnectedComponents}. 
By comparison, our algorithm is internally deterministic, which allows it to be used as a sub-component for e.g. MST algorithms.

Anderson et al.~\cite{anderson2020work} considered incrementally maintaining a spanning tree in a ``sliding window'' model, where edges may appear and disappear over time, which is different from ours. 
Our results should also extend to analyzing alternative ways of parallelizing sequential iterative algorithms, such as by defining task priorities and executing them via a relaxed priority queue~\cite{alistarh2018relaxed, postnikova2022multi}, in which case it can upper bound the amount of wasted work due to out-of-order execution; we leave this analysis for future work.

We also mention existing work on efficient \emph{concurrent} variants of Union-Find algorithms, which modify the original linking approaches to employ atomic operations allowing  concurrent access to the Union-Find data structure~\cite{JayantiTB19, AlistarhFK19, JayantiT21}. Our work is only partly related, as it considers a completely different parallelization approach, with different metrics and progress guarantees. 
Specifically, the above line of work considers an asynchronous shared-memory model with atomic operations, studying total step complexity. By contrast, we consider a standard parallel model, and study classic notions of task parallelism such as dependency depth and collisions.

\section{Notation and Preliminaries}

Arguably, the most common data structure to address the Disjoint-Set Union / Union-Find problem is the \emph{compressed forest}~\cite{Tarjan1975}. Here, nodes / set elements form directed trees with roots used as \emph{representatives} of their sets. As a result, checking whether two elements are in the same set can be implemented by following \emph{parent links} of these elements in the directed trees, finding the root representatives, and comparing them. Merging two sets means adding a link from one root to another one.

Algorithms based on \emph{compressed forest} data structures differ in two key ingredients. The first is the \emph{linking technique}: by choosing which root becomes the common root when merging two trees, it is possible to limit tree depth. Three popular linking strategies are \emph{linking by size}---linking the smaller tree to the larger one, \emph{linking by rank}---similar to \emph{linking by size} but tree sizes are approximated, and \emph{linking by random priorities}---the root with lower random priority is linked to the root with higher priority. All three linking strategies achieve $\O(\log n)$ maximum tree depth~\cite{Tarjan1975, Tarjan2014}, though in the case of \emph{linking by random priorities} this is in expectation. The second key component is \emph{path compaction}---traversed paths in trees are shortened after every root search by replacing parents of every visited node with nodes higher in the tree. When combined with any of the linking techniques, this results in $\O(\alpha(n))$ amortized time per operation, where $\alpha(n)$ is the inverse-Ackermann function.

We will associate \texttt{unite(u, v)} operations with $(u, v)$ edges in the union graph $G = (V, E)$, where vertices correspond to Union-Find elements. These edges can also be viewed as tasks; to execute the task means to unite sets corresponding to edge endpoints. This high-level approach allows us to study union-find properties for different graph structures.

\subsection{Definition of Collisions}

The key definition for analysing parallel and concurrent dynamic graph algorithms is that of an \emph{edge collision}. Intuitively, two edges \emph{collide} when they cannot be processed in parallel. To define what this means for the Union-Find problem, consider Listing~\ref{listing:union-find}. When joining sets (i.e., trees), the classical Union-Find implementation compares their sizes and then links the root of the smaller one to the root of the larger tree (lines~\ref{listing:union:parent-write1} and \ref{listing:union:parent-write2}). It is easy to see that concurrent links of the same root to different trees result in the loss of a link, and thus, are incorrect. In practice, this is where concurrent Union-Find algorithms employ synchronization primitives~(e.g., \texttt{Compare-And-Set} in \cite{JayantiT21}). 
We say that two edges \emph{collide} when they both connect different components and share the ``smaller'' component (the exact order of the components is determined by the linking strategy). Note that size updates in lines~\ref{listing:union:size-write1} and \ref{listing:union:size-write2} are much less crucial, since they commute and can be implemented with simple atomic operations such as \texttt{fetch-and-add}.

\begin{figureAsListing}
\begin{lstlisting}
bool union(u, v):
  u := find_root(u)
  v := find_root(v)
  if u == v: return false // Already in the same component
  if u.size < v.size: // Which component is larger?
  <@\indentrule@>  u.parent = v <@\label{listing:union:parent-write1}@> // Link u to v
  <@\indentrule@>  v.size += u.size <@\label{listing:union:size-write1}@>
  else:
  <@\indentrule@>  v.parent = u <@\label{listing:union:parent-write2}@> // Link v to u
  <@\indentrule@>  u.size += v.size <@\label{listing:union:size-write2}@>
  return true
\end{lstlisting}
\caption{\texttt{union} operation in the Union-Find data structure with the \emph{linking by size} technique. }
\label{listing:union-find}
\end{figureAsListing}

\begin{figure}
\centering
\includegraphics[width=0.9\linewidth]{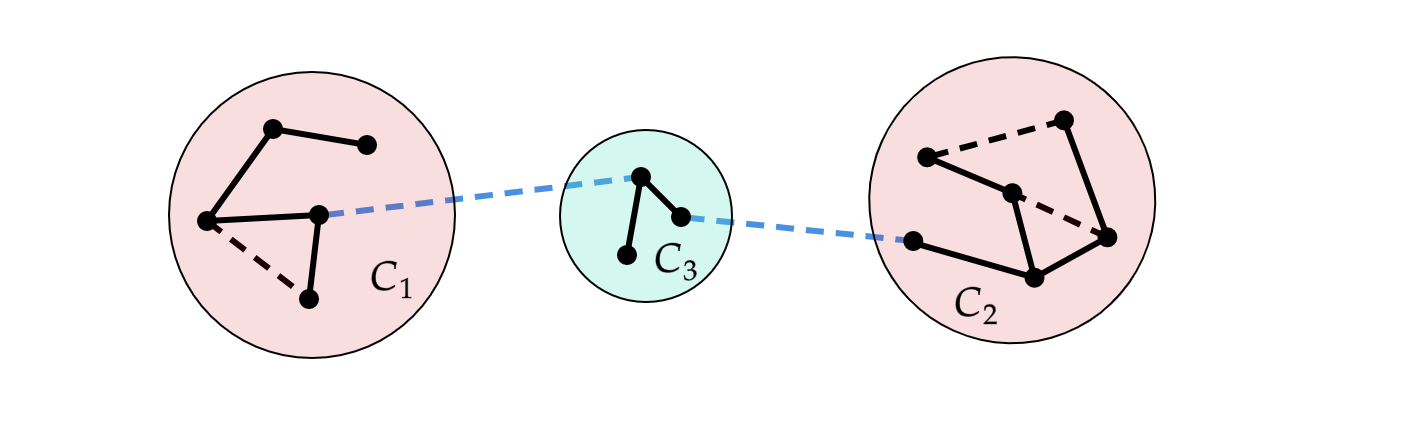}
\caption{An example of two edges colliding in the case of linking by size. The black solid edges are the ones already inserted. The dashed edges are yet to be added to the data structure. Since $C_3$ is the smaller component, the blue edges will try to write to the same memory location (the root of $C_3$), i.e., they have a collision and cannot be added to the data structure in parallel.}
\label{figure:collision_figure}
\end{figure}

One may ask whether the definition of collision can be simplified to requiring that ``two edges share an adjacent tree.'' This change will not make data structure implementations simpler but will simplify the analysis. Unfortunately, this  dramatically increases the number of collisions. Consider Erdos-R{\'e}nyi random graphs~\cite{ErdosRenyi} with the number of already inserted edges between $(\frac 1 2 + \varepsilon)|V|$ and $(1-\varepsilon)|V|$, where $\varepsilon$ is a constant strictly between $0$ and $\frac 1 4$. It is known that the largest component in this case is of size $\Omega(|V|^{2/3})$, and all other components are of sizes $\O(\log |V|)$, w.h.p.~\cite{ErdosRenyi}.  As a result, the probability of \emph{simplified} collision for two random edges is $\Omega(|V|^{-1/3})$, while in the case of our definition it is $\O\left(\frac {\log |V|} {|V|}\right)$. In fact, this suggests that for our definition the total number of collisions for random graphs is polylogarithmic, whereas for the simplified definition it is polynomial.

\subsection{The Random Process}
\label{subsection:random_process_definition}

Similar to prior work, e.g.~\cite{GreedyMISAreParallel, InternallyDeterministicParallel}, it is convenient to assume that the operations of the Union-Find algorithm are executed sequentially, i.e., edges are added one by one to the union graph $G$ in some order. This may not be the case in the actual execution, since parallelism allows to add several edges at once, although, in our parallel algorithm, to guarantee deterministic computations, these edges will still be from the prefix of the unprocessed edge sequence.  We call ``step $t$'' the moment of time when the $t$-th edge has been added to the data structure. Specifically, we denote the $t$-th added edge by $\varepsilon_t$ and the current set of inserted edges is $E_t = \{ \varepsilon_1, \cdots, \varepsilon_t \}$. Edges that are yet to be added (i.e., edges from $E\setminus E_t$) are called \emph{active}. In each step, we assume in our analysis that all present connected components are enumerated in the order defined by the used linking strategy (e.g., in order of increasing sizes when \emph{linking by size} is used) from $1$ to $C_t$. Finally, we call $m_i^t$ the number of active edges that connect the component of index $i$ to a connected component with a strictly greater index (i.e., to a ``larger'' component) in step $t$. Notice that these $m_i^t$ edges pairwise collide according to our definition of collision.

As in prior works analysing graph algorithms~\cite{ParallelismInRandomized,ParallelRandomPermutation,GreedyMISAreParallel}, we assume that the order of edges is uniformly random. Otherwise, there can be examples with $\Omega(|E|)$ collisions among subsequent pairs of edges in the sequence, resulting in no potential for parallelism. We will prove that the expected number of collisions among two random \emph{active} edges is small. More formally, let $X_t$ be an indicator random variable for the event ``two distinct active edges (i.e. from $E\setminus E_t$) chosen uniformly at random in step $t$   collide.'' Let the probability of collision for two random distinct active edges in step $t$ be $p_t = \E [X_t]$. 
% Note that this expectation is both over the choice of two random edges and random order of all edges. 
We will use conditional expectation $\E [X_t ~|~ \varepsilon_1, \ldots, \varepsilon_t]$ when previously added edges are known ($\varepsilon_t$ is a random variable when the order of edges is random).

The next theorem is the core of our performance analysis in Section~\ref{section:union-find-algorithms}, and is proven in Section~\ref{section:collision-analysis}. It states that the sum of probabilities of collisions over all steps is polylogarithmic. 

\begin{reptheorem}{theorem:sum_pt}
For a uniform random edge ordering, and any linking Union-Find strategy that bounds union-find forest depth to $\O(\log |V|)$, it holds that $\sum\limits_{t=0}^{|E|-1} p_t = \O(\log |V| \cdot \log |E|)$, where $|V|$ is the number of vertices and $|E|$ is the number of edges.
\end{reptheorem}

\subsection{Model of Parallelism}
\label{subsection:preliminaries-parallel-model}
For our parallel algorithm (Subsection~\ref{subsection:parallel-model}), we assume concurrent-reads concurrent-write parallel random-access machine (CRCW PRAM) with \emph{priority write}. Priority write (\texttt{write\_min}) is an atomic instruction that writes an input value to some memory address only if it is smaller than the current value at that location. This atomic instruction is widely used in parallel algorithms, especially for the Connected Components and Spanning Tree problems, and can be efficiently implemented~\cite{ParallelConnectedComponents, ParallelismInRandomized}.

We analyse our parallel algorithms using the standard work-depth approach. The \emph{work} of a parallel algorithm is the total number of executed instructions. The \emph{parallel depth} (or \emph{span}) is the length of the longest computational dependency chain.

\section{Collision Analysis}
\label{section:collision-analysis}

The goal of this section is to introduce an argument for bounding the number of collisions in the case of random edge ordering. First of all, we will prove the following lemma about collision probability when all added edges are known. Please recall that all notation has been defined in Subsection~\ref{subsection:random_process_definition}.

\begin{lemma} \label{lemma:pt-equals}
    For any $t \geq 0$, the collision probability in step $t$ for fixed $\varepsilon_1, \ldots, \varepsilon_t$ is 
    $$\E [X_t ~|~ \varepsilon_1, \ldots, \varepsilon_t] = \sum\limits_{i=1}^{C_t} \frac{m_i^t} {|E|-t} \cdot \frac {m_i^t-1}{|E|-t-1}$$
\end{lemma}
\begin{proof}
	The probability that we select an edge from component $i$ to a larger component is the number of such edges, denoted by $m_i^t$, divided by the number of remaining edges in step $t$, which is $|E|-t$. 
	The only edges that can cause a collision with the chosen edge are other edges connecting component $i$ to a larger component. Consequently, the probability that another random edge causes a collision with regard to selected edge is $\frac{m_i^t-1}{|E|-t-1}$, and therefore, the probability that two random edges collide is $\sum\limits_{i=1}^{C_t} \frac{m_i^t} {|E|-t} \cdot \frac {m_i^t-1}{|E|-t-1}$. 
\end{proof}

\paragraph{Overview} A natural proof strategy would be to directly bound  $p_t$. However, the value of $p_t$ can heavily depend on the structure of graph $G$. For example, it follows by simple calculation that for a star graph, $p_0$ is at least some constant greater than zero, while for a cycle graph $p_0=o(1)$, but $p_{|E|-3}$ is a non-zero constant. Instead, we prove that $p_t$ cannot be high for all possible $t$ by employing a type of amortized analysis. Specifically, we present a metric which increases considerably when $p_t$ is high, and is bounded at the same time. As a result, this will help us to bound $ \sum\limits_{t=0}^{|E|-1} p_t$. 

Let $depth_t(u)$ be the depth of vertex $u$ in the \emph{uncompressed} Union-Find forest in step $t$, where ``uncompressed'' means that depth is counted ignoring any previous path compaction. Most reasonable linking techniques guarantee that for any vertex $u$, we have that $depth_t(u) \leq C\cdot \log |V|$ for some constant $C>0$. We define our potential function as follows:

\begin{definition} \label{definition:potential}
For each edge $(u, v) \in E$, we define its rank at time $t$ as $\phi_t(u, v) = depth_t(u) + depth_t(v)$. The potential function at $t$ will be $\Phi_t = \frac{|E|}{|E| - t} \sum\limits_{e \in E \setminus E_t} \phi_t(e) + \sum\limits_{i=1}^t \frac{|E|}{|E| - i} \cdot \phi_i(\varepsilon_i)$, $0 \leq t \leq |E| - 1$.
\end{definition}

% TODO: say about how we use it only after edges are processed

In more detail, the potential function is a sum over all edges of their ranks, but the ranks are multiplied by a constantly growing factor $\frac {|E|} {|E| - t}$, and the rank $\phi_i$ and its multiplier for edge $\varepsilon_i$ are ``frozen'' at time $i$ when the edge was processed.

Initially, we have $\Phi_0 = 0$, as the graph is empty, and thus, the ranks of all edges are zero. Let us observe how this potential is expected to change after each step, given that all remaining edges have equal probability to be added to the data structure. That is, we follow the sum $\E [\Phi_1 - \Phi_0] + \E[\Phi_2 - \Phi_1 ~|~ \varepsilon_1] + \ldots + \E[\Phi_{|E|-1} - \Phi_{|E|-2} ~|~ \varepsilon_1, \varepsilon_2, \ldots \varepsilon_{|E|-2}]$. 
Note that, by the law of total expectation, $\E\left[\E [\Phi_1 - \Phi_0] + \ldots + \E[\Phi_{|E|-1} - \Phi_{|E|-2} ~|~ \varepsilon_1, \varepsilon_2, \ldots \varepsilon_{|E|-2}]\right]$ $= \sum\limits_{t=1}^{|E|-1} \E [\Phi_t - \Phi_{t-1}] = \E [\Phi_{|E|-1}]$. 
Next, we connect this potential with the probability of collision.

\begin{lemma}\label{lemma:delta-lowerbound}
$\E [\Phi_{t+1} - \Phi_t ~|~ \varepsilon_1, \ldots, \varepsilon_t] \geq |E| \cdot \sum\limits_{i=1}^{C_t} \frac {m_i^t} {|E| - t} \cdot  \frac {m_i^t} {|E| - t - 1}$.
\end{lemma}

\begin{proof}
 Assume that in step $t+1$, the newly added edge $\varepsilon_{t+1}$ is $(u, v) \in E$. If the edge is internal, i.e., connects nodes in the same component, the edge does not affect connectivity or edge ranks. The potential function still increases because of the increasing multiplier $\frac {|E|} {|E|-t}$, but we use a trivial lower bound of $0$ for this increase. If the edge is external, i.e., connects two different connected components, assume that the index of the smallest component is $i$. Then, the number of edges connecting this component to larger components is $m_i^t$, so the probability that this component is chosen as the smaller one is $\frac {m_i^t} {|E|-t}$. Due to the fact that the root of this component will be linked to another component, ranks of all active edges adjacent to it are increased by at least $1$. We can lower bound the number of such edges by $m_i^t$. This means that, with probability at least $\frac {m_i^t} {|E|-t}$, component $i$ is chosen and then the sum of active edge ranks is increased by at least $m_i^t$. Finally, taking into consideration the multiplier of the potential function and summing over all components, we get the required inequality
 
 \begin{gather*}
 \E [\Phi_{t+1} - \Phi_t ~|~ \varepsilon_1 \ldots \varepsilon_t] \geq \frac {|E|} {|E|-t-1} \cdot \sum\limits_{i=1}^{C_t} m_i^t \cdot \frac {m_i^t} {|E| - t} = |E| \cdot \sum\limits_{i=1}^{C_t} \frac {m_i^t} {|E| - t} \cdot  \frac {m_i^t} {|E| - t - 1}. 
 \end{gather*}

\end{proof}

\begin{lemma} \label{lemma:potential-collision-connection}
 The increment of the potential at each step $t \geq 0$ satisfies
$\E [\Phi_{t+1} - \Phi_t ~|~ \varepsilon_1, \ldots, \varepsilon_t] \geq |E| \cdot \E [X_t ~|~ \varepsilon_1, \ldots, \varepsilon_t].$
\end{lemma}
\begin{proof}
Follows by combining Lemma~\ref{lemma:pt-equals} with the  inequality in Lemma~\ref{lemma:delta-lowerbound}. Note that these expectations are different {--} the left one is over the choice of $\varepsilon_{t+1}$, while the right one is over the choice of two random active edges. %Inequality of conditional expectations in this lemma in fact means that this inequality is true for any values of $\varepsilon_1 \ldots \varepsilon_t$. 
\end{proof}

\begin{corollary} \label{corollary:potential-collision-connection}
$\E [\Phi_{|E|-1}] \geq |E| \cdot \sum\limits_{t=0}^{|E|-2} p_t$.
\end{corollary}
\begin{proof}
Summing up inequality from Lemma~\ref{lemma:potential-collision-connection} over all steps and taking expectation of the sum, we get the following:
\begin{gather*}
    \E \left[\sum\limits_{t=0}^{|E|-2} \E [\Phi_{t+1} - \Phi_t ~|~ \varepsilon_1, \ldots, \varepsilon_t]\right] \geq |E| \cdot \E \left[\sum\limits_{t=0}^{|E|-2} \E [X_t ~|~ \varepsilon_1, \ldots, \varepsilon_t]\right].
\end{gather*}
By the law of total expectation, this can be simplified to the next inequality:
\begin{gather*}
    \E \left[\sum\limits_{t=0}^{|E|-2} \Phi_{t+1} - \Phi_t \right] \geq |E| \cdot \E \left[\sum\limits_{t=0}^{|E|-2} X_t\right]
\end{gather*}
Finally, using the definition of $p_t$ and reducing the left sum, we deduce that:
\begin{gather*}
    \E [\Phi_{|E|-1}] = \E [\Phi_{|E|-1} - \Phi_0]  \geq |E| \cdot \sum\limits_{t=0}^{|E|-2} p_t.
\end{gather*}

\end{proof}

This corollary means that in order to upper bound the smoothed collision probabilities $\sum\limits_t p_t$, we can bound the potential function instead. We will do this in the next lemma. Note that for randomized linking strategies (e.g., \emph{linking by random priorities}) the next bound holds in expectation over the  randomness used in linking.

\begin{lemma} \label{lemma:potential-bound}
    For any linking strategy bounding maximal depth to $\O(\log |V|)$, we have $\Phi_{|E|-1} = \O({|E|} \cdot \log |V| \cdot \log |E|)$.
\end{lemma}
\begin{proof}

When forest depth is $\O(\log |V|)$, it is easy to see that every edge rank is $\O(\log |V|)$, and thus, the sum of edge ranks is $\O(|E| \cdot \log |V|)$. The problem is the $\frac {|E|} {|E| - t}$ edge rank multiplier, which ranges from $1$ to $|E|$. That is why we ``freezed'' ranks of edges upon their processing.

Observe the edge rank multipliers for edges added in different steps. For the first $\frac {|E|} {2}$ edges, this multiplier is at most $\frac {|E|} {|E| - |E|/2} \leq 2$. Similarly, for the next $\frac {|E|} {4}$ edges, this multiplier is at most $\frac {|E|} {|E|-3|E|/4} \leq 4$. By continuing this progression, we get the following for forest depth $\leq C \cdot \log |V|$:

\begin{gather*}
    \Phi_{|E|-1} \leq \frac {|E|} 2 \cdot 2 \cdot C \cdot \log |V| \cdot 2 + \frac {|E|} 4 \cdot 4 \cdot C \cdot \log |V| \cdot 2 + \ldots 
    =\\= \sum\limits_{i=0}^{\log_2 |E|} |E| \cdot C \cdot \log |V| \cdot 2 = \O(|E| \cdot \log |V| \cdot \log |E|).
\end{gather*}
\end{proof}

\begin{theorem}
\label{theorem:sum_pt}
The sum of collision probabilities satisfies $\sum\limits_{t=0}^{|E|-1} p_t = \O(\log |V| \cdot \log |E|)$ given a random order of edges and any linking Union-Find strategy that bounds forest depth to $\O(\log |V|)$.
\end{theorem}
\begin{proof}
Combining Corollary~\ref{corollary:potential-collision-connection} and Lemma~\ref{lemma:potential-bound}, we conclude $\sum\limits_{t=0}^{|E|-2} p_t \leq  \frac {1} {|E|} \cdot \Phi_{|E|-1} = \O(\log |V| \cdot \log |E|)$.
\end{proof}

Surprisingly, the proof of Theorem~\ref{theorem:sum_pt} shows  $\sum\limits_{t=0}^{|E|/2} p_t \leq C \cdot \log |V|$, i.e., for the first half of edges the expected number of collisions is just $\O(\log |V|)$, which may mean that the factor of $\log |E|$ is an artefact of the current analysis.

\section{Work-Efficient Union-Find}
\label{section:union-find-algorithms}

We start by showing how to apply Theorem~\ref{theorem:sum_pt} to some practical Union-Find algorithms in various computation models. Specifically, we will speak about incremental dynamic connectivity and minimum spanning tree problems.

\subsection{Bounding Contention}
\label{subsection:concurrent-model}

The discussion so far assumed an idealized sequential execution. 
We now wish to analyse concurrent Union-Find algorithms, which can be implemented in practice via locks, hardware transactions, or lock-free primitives~\cite{AlistarhFK19, ConnectIt}, using a generalized concurrency model. In all these implementations, extra work comes from \emph{write contention}---situation when several threads try to modify the same memory location simultaneously. For locks, write contention means the need of waiting until another thread finishes modifying the required memory location, while for hardware transactions and lock-free primitives, it means retries of operations. Moreover, in practice, write contention causes additional L3 cache misses in NUMA (Non-Uniform Memory Access) computer architectures. This is why we will focus on bounding the probability that threads experience write contention when attempting to process edges in parallel. 

More formally, our basic concurrency model is: a set of $T \geq 2$ threads execute synchronously in iterations, where in each iteration each thread picks a remaining edge uniformly at random, and inserts it into the Union-Find data structure. 
The threads may pick colliding edges, in which case we register a \emph{write contention} event; otherwise, the edges are processed without contention. Either way, we assume that the edges are processed.
Our goal is to bound the expected total number of write contention events between threads when processing all edges. In our case, write contention between threads means collision between their edges according to our definition of collision. 

\begin{theorem} \label{theorem:concurrent-two-threads}
The expected number of write contention events for two threads in this concurrent model is bounded by $\O(\log |V| \cdot \log |E|)$. 
\end{theorem}
\begin{proof}
Since edges are inserted randomly into  the data structure, this algorithm is almost the same as the random process described in Subsection~\ref{subsection:random_process_definition}. Specifically, iteration $t$ in this algorithm corresponds to step $2t$ in the random process because two new edges are added every iteration. Moreover, the probability of write contention at iteration $t$ is exactly $p_{2t}$ (see the definitions in Subsection~\ref{subsection:random_process_definition}). So, the total number of write contention events is $\sum\limits_{t=0}^{|E|/2} p_{2t} \leq  \sum\limits_{t=0}^{|E|} p_t = \O(\log |V| \cdot \log |E|)$. The last inequality is from Theorem~\ref{theorem:sum_pt}, stated in the end of the previous Section.
\end{proof}

\begin{theorem} \label{theorem:concurrent-T-threads}
The expected number of write contention events for $T$ threads in this concurrent  model is $\O(T^2 \cdot \log |V| \cdot \log |E|)$.
\end{theorem}
\begin{proof}
The number of possible pairwise contentions every iteration is $\Theta(T^2)$, so by linearity of expectation, the expected number of collisions at iteration $t$ is $\O(T^2 \cdot p_{T\cdot t})$. Then, the total number of contention events is bounded by $\sum\limits_{t=0}^{|E|/T} T^2 \cdot  p_{T\cdot t} \leq  \sum\limits_{t=0}^{|E|} T^2\cdot p_t = \O(T^2 \cdot \log |V| \cdot \log |E|)$.
\end{proof}

The last two results show that the number of write contention events depends on the size of the graph only polylogarithmically, and thus, may be dominated by other costs of processing a graph. This suggests fairly low contention cost for implementing Union-Find via most concurrency primitives (locks, atomic operations, transactions), which was also previously shown in practice~\cite{ConnectIt}. 

\subsection{Parallel Algorithm}
\label{subsection:parallel-model}

Our main goal is to design a parallel iterative algorithm for Union-Find that is work-efficient and internally deterministic. 
While work-efficient parallel algorithms are well-known for many problems, relatively few algorithms have the second property~\cite{InternallyDeterministicParallel,ParallelismInRandomized,ParallelRandomPermutation}. 
% Intuitively, internally deterministic parallel algorithms are the algorithms that always do the same work as the sequential ones. 
Internal determinism is particularly interesting for Union-Find, since running an internally-deterministic parallel Union-Find algorithm on a sorted sequence of edges is similar to Kruskal's algorithm~\cite{KruscalMST} and results in a minimal spanning tree. Blelloch et al.~\cite{InternallyDeterministicParallel} present a practical algorithm for Union-Find but without a theoretical analysis. Their algorithm is deterministic but \emph{not internally deterministic}, as its output (i.e., edges used for connectivity) depends on some predefined parameter and may diverge from the output of the sequential algorithm.

 We aim to close these gaps by providing a practically-efficient algorithm which is provably work-efficient, and \emph{internally deterministic}.  Specifically, the algorithm uses concurrent-reads concurrent-write (CRCW) PRAM with \emph{priority write} (see Subsection~\ref{subsection:preliminaries-parallel-model}).

\paragraph{Deterministic Reservations} 
We build on the algorithm of Blelloch et al.~\cite{InternallyDeterministicParallel}, which uses the \emph{deterministic reservations} approach. Specifically, in each iteration, their algorithm considers the prefix of remaining tasks (i.e. edges) of size $S$ and then proceeds in two phases. In the first phase, these tasks do \emph{reservations} of memory locations they want to change by using \emph{priority write}. In the second phase, all tasks in the prefix that succeeded in their reservations are executed. For the Union-Find problem, in the first phase tasks reserve the smaller adjacent connected component for each edge, by making a priority write in its root and then link roots of these smaller components to other components in the case of successful reservation. This behaviour follows our definition of collision: the algorithm adds all edges in the prefix that do not collide with any preceding edges.

There are several challenges when making this algorithm  both work-efficient and internally deterministic. First, Blelloch et al. do not analyze how often collisions happen and how many iterations the algorithm requires. 
Furthermore, parallelism in their algorithm violates linking strategy properties, which may make their algorithm not work-efficient. For example, on a directed path, their algorithm can add all edges in one iteration, but the depth of the resulting Union-Find forest will be linear. Last but not least, the algorithm is not internally deterministic: there are cases where the sequential algorithm uses some edge for connectivity (i.e., $\texttt{union}(u,v)$ returns \texttt{true}), but this algorithm replaces it with an edge located later in the sequence. Blelloch et al. suggest a modification of their algorithm to solve Spanning Tree, which does make it internally deterministic but it requires each edge to reserve \emph{both} of its endpoints. As a result, for random or star graphs, few operations will succeed in every iteration, since most operations will try to reserve the largest connected component. 

\begin{figure}[ht]
\centering
\includegraphics[width=0.6\linewidth]{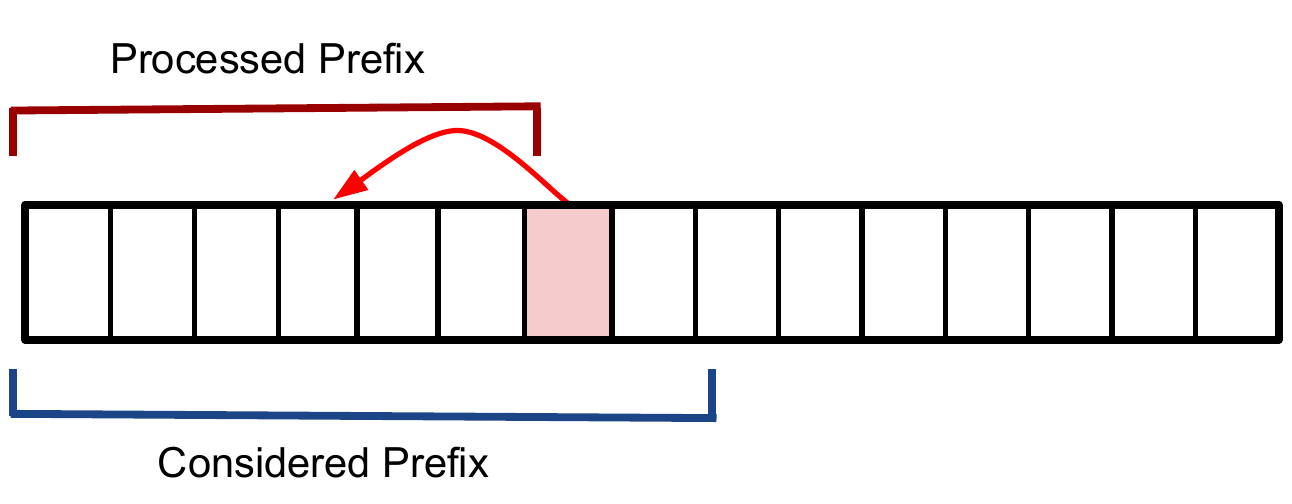}
\caption{An example how our parallel algorithm processes edges. It scans some prefix of the remaining edge sequence and does deterministic reservations. If there is a collision, i.e., for some edge, deterministic reservation failed because of another edge, the algorithm processes not all considered prefix, but only its part before the first collision.}
\label{figure:parallel-model}
\end{figure}

\paragraph{Internally-Deterministic Algorithm} 
One key difference between the algorithm of Blelloch et al. and the one we are analyzing is that we do not allow the execution of a task unless all tasks before it will be executed by the end of the current iteration (Figure~\ref{figure:parallel-model}). In other words, we execute all tasks in the prefix until a task with failed reservation, i.e., an edge with a collision with one of the previous edges. This modification addresses the latter problem in the algorithm of Blelloch et al. by making it internally deterministic. 
This appears to make the algorithm less efficient; yet, we will prove that collisions are relatively rare, so the modification does not have significant impact on  performance. Similarly to other iterative algorithms~\cite{ParallelismInRandomized,ParallelRandomPermutation}, we assume that the order of edges is uniformly random.

\begin{figureAsListing}
\begin{lstlisting}
void make_reservations(edges, l, r):
  parallel for t <@$\in$@> [l, r):
    (u, v) := edges[t] // alias for edge endpoints
    <@\label{listing:window:find1}@>u_root = find_root(u)
    <@\label{listing:window:find2}@>v_root = find_root(v)
    if u_root == v_root:
      continue // already in the same component
    if u_root < v_root: // linking strategy
      write_min(u_root.reservation, t) 
    else:
      write_min(v_root.reservation, t)

void union_find(edges):
  i := 0
  while i < |edges|:
    <@\label{listing:window:reservations}@>make_reservations(edges, i, i + S):
    <@\label{listing:window:stop}@>stop := min edge id with failed reservation or i + S 
    sz := stop - i // max size without collisions
    <@\label{listing:window:link}@>Parallel-Link-All(i, i + sz) 
    i += sz
\end{lstlisting}
\caption{Parallel algorithm for Union-Find. In each iteration, a prefix of size $S$ is being processed in parallel unless there is a write collision among its tasks. \texttt{Parallel-Link-All} was presented in the work-efficient Union-Find~\cite{WorkEfficientUnionFind}. It groups roots according to connected components and then unites them via a recursive divide-and-conquer strategy.}
\label{listing:window-union-find}
\end{figureAsListing}

Our parallel algorithm is presented in Listing~\ref{listing:window-union-find}. It starts the same way as the algorithm of Blelloch et al.~\cite{InternallyDeterministicParallel} by making \emph{deterministic reservations} in the unprocessed prefix of size $S$ (Line~\ref{listing:window:reservations}). Specifically, the algorithm finds the roots of the corresponding trees for each edge and tries to reserve the smaller one according to the used linking strategy. Then, it checks in parallel whether all reservations are successful and, if not, it finds the first unsuccessful reservation using \texttt{write\_min} instruction (Line~\ref{listing:window:stop}). Finally, the algorithm completes all edge additions in the prefix until the first failed reservation (Line~\ref{listing:window:link}). In the simplest case, this means just linking the root that was designated as ``smaller'' in the previous step to the larger one. However, we will propose another approach that maintains all the guarantees of the linking techniques on forest depth. This process is repeated until all edges are processed.

\paragraph{Asymptotic Optimizations} A naive implementation would spend $\O(S \cdot depth)$ total work and has $\O(depth)$ span for each iteration, where $depth$ is the Union-Find forest depth; due to parallelism, even with linking techniques, the forest depth potentially can be up to linear in the simplest algorithm. However, incremental improvements can address this and improve total work to an optimal $\O(S \cdot \alpha(|V|))$ while keeping the span polylogarithmic, where $\alpha(\cdot)$ is the inverse Ackermann function. Both adjustments were first proposed by Simsiri et al.~\cite{WorkEfficientUnionFind} in their work-efficient batched parallel Union-Find algorithm. First of all, \texttt{Parallel-Link-All} can preserve forest depth bounds if a linear-work connected components parallel algorithm is used to group all connected components that will form one component~\cite{WorkEfficientUnionFind}. Then, all grouped connected components can be merged in parallel by applying a recursive divide-and-conquer approach. In particular, when merging two Union-Find trees, any linking strategy can be used, since all information about these trees has been calculated recursively. This \texttt{Parallel-Link-All} takes linear work and polylogarithmic span. The second optimization ensures that parallelism in path compaction does not increase total work. Specifically, in lines~\ref{listing:window:find1} and \ref{listing:window:find2} the same bulk-parallel approach as in~\cite{WorkEfficientUnionFind} is employed. Roughly, the algorithm acts like a parallel BFS (Breadth-First Search) that starts in vertices,  for which we want to find the root, and ascends wave-by-wave until it finds all roots. Then it traverses all nodes again and re-links them directly to the roots. These adjustments allow Union-Find to use both path compaction and linking strategies, and as a result, the average work per edge is the optimal $\O(\alpha(|V|))$. 

\subsection{Algorithm Analysis}

\begin{theorem}
The Union-Find algorithm in Listing~\ref{listing:window-union-find} is internally deterministic.
\end{theorem}
\begin{proof}
We proceed by contradiction. Assume that the sequential algorithm uses a different set of edges to form connected components. Consider the first edge $(u, v)$ for which the result differs. It was either taken by the presented algorithm and not taken by the sequential one or vice versa.

\emph{Case 1}: This edge was used by our algorithm but not by the sequential one. The fact that it was not used by the sequential algorithm means that it, together with the previous edges, forms a cycle. As it was the first difference, our algorithm acted incorrectly and created a cycle in the Union-Find forest. However, linking techniques impose a total order on all connected components and the algorithm links smaller roots to larger ones. The cycle should have all edges directed in the same direction since in one iteration only one outgoing edge can be added for each connected component. This means that obtaining a cycle is impossible, because in a directed cycle at least one edge (link) contradicts the total order at the start of iteration and our algorithm does not allow this.

\emph{Case 2}: $(u, v)$ was used by the sequential algorithm and not by our algorithm. Consider the iteration when the edge was processed. We know that, at each iteration, all edges inside the processed prefix either tried to make a reservation and succeeded, or did not try because the vertices are already connected. We know the former is not true, because by assumption, $(u,v)$ was not used, so the latter should be true, which means that $u$ and $v$ were already connected before the iteration. However, we know that the edges added before the iteration cannot connect $u$ and $v$, otherwise the sequential algorithm would make a cycle. We have obtained a contradiction. 
\end{proof}

After proving that our algorithm is indeed internally-deterministic, what is left is to bound the number of its iterations. 

\begin{theorem} \label{theorem:window-iterations}
The expected number of iterations of the Union-Find algorithm in Listing~\ref{listing:window-union-find} is $\left\lceil \frac {|E|} {S} \right\rceil + \O(S^2 \cdot \log |V| \cdot \log |E|)$, if the order of edges is uniformly random.
\end{theorem}
\begin{proof}
It is easy to see that the minimum number of required iterations in absence of collisions is $\left\lceil \frac {|E|} {S} \right\rceil$ , since at most $S$ edges are processed in each iteration. Every collision in the prefix can cause the algorithm to perform at most one more iteration. So, to bound the number of ``extra'' iterations we can bound the total number of collisions across all considered prefixes of unprocessed edges.

Let $k$ be an arbitrary index in the edge sequence. We want to bound the probability of collision inside the window denoted by positions $[k,k+S)$ if all edges before $k$ were inserted in the data structure. For this reason, we define $Y_k$ to be an indicator random variable for the event ``there is at least one collision inside window [k, k+S] if all edges before $k$ were inserted in the data structure''.
Since the edge sequence is random, the inserted edges and edges inside this window are random as well. For two fixed positions in the window the probability of their collision then is exactly $p_k = \E \left[\E [X_k ~|~ \varepsilon_1, \ldots, \varepsilon_k]\right] $ (see Subsection~\ref{subsection:random_process_definition}). By the union bound, the probability of a collision in the window $\E \left[\E [Y_k ~|~ \varepsilon_1, \ldots, \varepsilon_k]\right]$ is then $\leq S^2 \cdot p_k$, since there are less than $S^2$ pairs of possibly colliding positions in the window.

Let $I$ be the set of unprocessed prefix start positions (i.e., the set of all observed values of \texttt{i} in Listing~\ref{listing:window-union-find}). Similarly to the previous paragraph, we would want to say that the total expected number is bounded by $\sum\limits_{i \in I} S^2 \cdot p_i$. However, the content of $I$ has a correlation with the sequence of edges. For instance, if in some iteration there was a collision, we know that the unprocessed prefix of the next iteration starts with an edge that had a collision with one of the edges before. This means that the order of edges in the unprocessed tail is not uniformly random anymore and we cannot use the definitions from Subsection~\ref{subsection:random_process_definition}.

This issue has two possible solutions. The first one requires a small modification to Listing~\ref{listing:window-union-find}: instead of processing the window $[i, stop)$, the algorithm will process this window and then will separately process the task at position $stop$. This modification does not asymptotically increase the total work or the span of an iteration. However, the separately processed edge at position $stop$ removes the correlation between the order of remaining edges and the start position of the next prefix. This is because the edges strictly after position $stop$ do not affect the current iteration, and thus, all possible permutations of the tail continue to be equiprobable. 

The second solution does not require modifying Listing~\ref{listing:window-union-find} but complicates the analysis. Consider the toy algorithm in Listing~\ref{listing:toy-window-union-find}. It adds edges one by one and at the same time counts the number of pairwise collisions in all possible windows of size $S$. Note that this do not have the same problem as the algorithm in Listing~\ref{listing:window-union-find} and all permutations of the tail of edges continue to be equiprobable, since the algorithm tries \emph{all} windows of size $S$. This allows us to say that the probability of a collision in window $i$ is $\E [Y_i]$, and thus, by the union bound, the total expected number of collisions is $\E \left[ \sum\limits_{i=0}^{|E|-1} Y_i \right] \leq \sum\limits_{i = 0}^{|E|-1} S^2 \cdot p_t = \O(S^2 \cdot \log |V| \cdot \log |E|)$. The last equality is Theorem~\ref{theorem:sum_pt}. Finally, the last observation we have to make is that the number of collisions counted by the algorithm in Listing~\ref{listing:toy-window-union-find} is not less than the number of collisions that occurred in the algorithm in Listing~\ref{listing:window-union-find} due to the fact that all windows considered in the latter algorithm were also considered by the former algorithm.

\begin{figureAsListing}
\begin{lstlisting}
int calculate_collisions(edges):
  i := 0
  collisions := 0
  while i < |edges|:
    collision += #pairwise collisions in [i, i + S)
    Union-Find.union(edges[i]) // insert only one edge
    i += 1
  return collisions
\end{lstlisting}
\caption{Toy algorithm that helps to prove Theorem~\ref{theorem:window-iterations}. It calculates the number of collisions across all possible windows of size $S$.}
\label{listing:toy-window-union-find}
\end{figureAsListing}

Both approaches show that only $\sum\limits_{i = 0}^{|E|-1} S^2 \cdot p_t = \O(S^2 \cdot \log |V| \cdot \log |E|)$ expected collisions happen in the algorithm, and consequently, the number of extra iterations is $\O(S^2 \cdot \log |V| \cdot \log |E|)$.

\end{proof}

\begin{corollary} \label{corollary:window-work-efficient}
The proposed Union-Find algorithm is work-efficient in expectation when $S = \O(|E|^{\frac 1 3 - \varepsilon})$ for any constant $\varepsilon > 0$. 
\end{corollary}
\begin{proof}
We know that one iteration takes $\O(S \cdot \alpha(|V|))$ total work. As a result, the expected total work is the expected number of iterations multiplied by the iteration work, which is $\O(S \cdot \alpha(|V|) \cdot (\frac {|E|} S + S^2 \cdot \log |V| \cdot \log |E|)) = \O(|E| \cdot \alpha(|V|) + S^3 \cdot \log^3 |E|) = \O\left(|E| \cdot \alpha(|V|) + |E| \cdot \frac {\log^3 |E|} {|E|^{3\varepsilon}}\right) = \O(|E| \cdot \alpha(|V|))$. %The last transition is the result of asymptotic inequality $\log^3 |E| = o(|E|^{3 \varepsilon})$.
\end{proof}

\begin{corollary} \label{corollary:window-iterations}
The number of iterations for the presented Union-Find algorithm is $\O(|E|^{\frac 2 3} \cdot \log^{\frac 1 3} |E| \cdot \log^{\frac 1 3} |V|)$ in expectation when $S$ is chosen optimally. The expected parallel depth (span) of the algorithm is then $\O(|E|^{\frac 2 3} \cdot \emph{polylog}(|E|))$.
\end{corollary}
\begin{proof}
It is easy to prove by substitution that the asymptotic minimum of the equation in Theorem~\ref{theorem:window-iterations} is achieved when $S = \sqrt[3]{\frac {|E|} {\log |V| \cdot \log |E|}}$. The span directly follows from the bound on the number of iterations and the polylogarithmic span of one iteration.
\end{proof}

In addition, we observe that when the prefix size is $S$, up to $S$ processes can be leveraged by the algorithm. When combined with Corollary~\ref{corollary:window-work-efficient}, we get the following theorem:

\begin{theorem} \label{theorem:main-parallel-algorithm}
There exists an internally-deterministic parallel Union-Find algorithm for CRCW PRAM model with priority write that has $\O(|E| \cdot \alpha(|V|))$ expected total work on a randomly shuffled sequence of edges and scales perfectly for up to $\O(|E|^{\frac 1 3 - \varepsilon})$ processes.
\end{theorem}

It is worth mentioning that for dense graphs this total work becomes linear in the number of edges. This is because the sequential Union-Find data structure has another work bound of processing $|E|$ edges, which is $\O(|E| + |V| \cdot \lg^* (|V|))$, where $\lg^*$ is iterated logarithm. As a result, when $|E| > |V| \cdot \lg^* (|V|)$, this bound is better than $\O(|E| \cdot \alpha(|V|))$ and is just $\O(|E|)$.

\paragraph{Iteration Dependence Depth}
Another metric of interest is the \emph{iteration dependence depth}~\cite{ParallelismInRandomized}. 
Yet, the definition of this metric is not obvious for the Union-Find problem. 
A natural definition could be that \emph{there is a dependency between tasks if at some time there is a data race between them}. Following this definition, Corollary~\ref{corollary:window-iterations} means that the expected iteration dependency depth is $\O(|E|^{\frac 2 3} \cdot \log^{\frac 1 3} |E| \cdot \log^{\frac 1 3} |V|)$: by construction, there are no data races in the same processed prefix, which results in the iteration depth being bounded by the number of processed prefixes.  

A more general definition for dependency depth could be that exactly the same set of \texttt{union(u, v)} operations should succeed (i.e., return \texttt{true}). 
However, for this definition, the dependence depth is constant: failed \texttt{union(u, v)} operations depend on successful \texttt{union} operations that formed the path between $u$ and $v$, while successful operations do not have any dependencies. 

\subsection{Simple Extension to MST} 

Lastly, since the algorithm is internally deterministic, it works similarly to a parallel version of Kruskal's algorithm on a sorted sequence of edges. We believe this is the first direct parallelization of Kruskal's algorithm. 
For example, in the parallel MST algorithm of Katsigiannis et al. ~\cite{KruskalHelperThreads}, the main thread processes the whole edge sequence, while other threads are helping by checking and filtering out internal edges. In this case, the work of \emph{successful} \texttt{unite(u, v)} operations was still essentially sequential and done by the same thread. Moreover, in Blelloch et al. MST algorithm~\cite{InternallyDeterministicParallel}, few edges are processed on average in one iteration for random and star graphs. By contrast, we can parallelize this process here. The key restriction is that, to be provably efficient, our approach requires the sorted sequence of edges to be in fact a random shuffle of edges.

\begin{corollary}
There exists a MST algorithm for CRCW PRAM model with priority write that has $\O(|E| \cdot \alpha(|V|))$ expected total work, when random edge weights are generated independently and the edge sequence is already sorted, and it scales  for $T = \O(|E|^{\frac 1 3 - \varepsilon})$  parallel threads, for any constant $\varepsilon > 0$.
\end{corollary}

\section{Lower Bounds}

In the previous sections, we showed upper bounds for different metrics, leaving open the question of tightness of these bound. We now attempt to close this gap. Our proofs require the reasonable assumption that the Union-Find tie-breaker when there are isomorphic connected components (in particular, of the same rank or size) acts randomly, i.e., compares random priorities, since it does not have any additional information about the graph structure.

\subsection{Lower Bound on Collision Count} We start by providing a lower bound on the number of collisions (Theorem~\ref{theorem:sum_pt}). Specifically, we show that the poly-logarithmic bound we provided cannot be improved by more than a logarithmic factor.

\begin{theorem} \label{theorem:lower-bound-pt}
For a cycle graph, we have  $\sum\limits_{t=0}^{|E|-1} p_t = \Omega(\log |V|)$.
\end{theorem}
\begin{proof}

\begin{figure}
\centering
\includegraphics[width=0.23\linewidth]{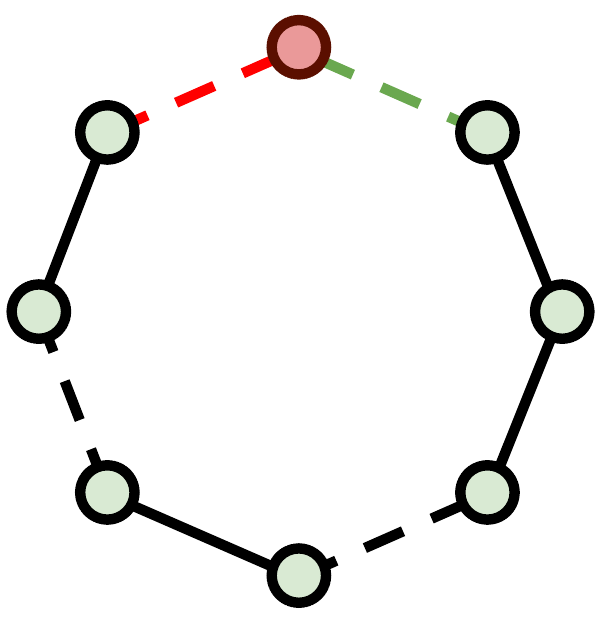}
%\vspace{-0.5em}
\caption{An example of Union-Find execution on a cycle graph. Solid lines symbolize already added edges, while dotted lines are edges that are not inserted yet. The red edge was randomly chosen, its clockwise next component is smaller than its neighbours, so the only edge this edge can have collision with it is the green edge.}
\label{figure:cycle-graph}
\end{figure}

Consider a cycle graph. In step $t$, there are exactly $|V|-t$ connected components and $|V|-t$ active edges. Note that since we added edges in a uniformly random order, any combination of $t$ edges from the cycle is equally probable to be picked. For every connected component, there are two neighboring components (i.e., components adjacent through an active edge) in any step $t$, when $t \leq |V| - 3$. An important observation is that if we contract every currently present connected component in a single vertex, the remaining graph is still a cycle. When we add a random edge, for the clockwise-next endpoint component there is a $\frac 1 3$ chance that it is smaller than both of its neighbouring components due to symmetry of the graph (unless there are two or less components in total, but this occurs only in the  last $3$ steps). If it is smaller than its neighbours, then there is exactly one edge (out of $|V|-t$ edges) in the graph that can cause a collision with the chosen one -- the edge between this endpoint and its other neighbouring component (Figure~\ref{figure:cycle-graph}). So, we get:

\begin{gather*}
\sum\limits_{t=0}^{|V|-1} p_t \geq  \sum\limits_{t=0}^{{ |V|-3}} p_t \geq \sum\limits_{t=0}^{|V|-3} \frac {1} {3} \cdot \frac 1 {|V|-t - 1} = \Theta(\log |V|). 
\end{gather*}
(the sum of Harmonic progression)
\end{proof}

This implies the following for the model in Section~\ref{subsection:concurrent-model}: 

\begin{theorem}
The expected write contention for two threads in the concurrent model is $\Omega(\log |V|)$ for cycle graphs.
\end{theorem}
\begin{proof}
As we know from Theorem~\ref{theorem:lower-bound-pt}, for a cycle graph $p_t \geq \frac 1 3 \cdot \frac 1 {|V| - t - 1}$. On the other hand, from Theorem~\ref{theorem:concurrent-two-threads}, the expected number of write contentions is $\sum\limits_{t=0}^{|E|/2} p_{2t}$. As a result, the required bound is the sum of harmonic series over every second element and is still $\Theta(\log |V|)$.
\end{proof}

\subsection{Lower Bound on Number of Iterations}
A significantly more involved analysis is required for the parallel algorithm from Section~\ref{subsection:parallel-model}. The challenge is that we cannot use the symmetry argument to show that the probability of collision for an adjacent pair of edges is $\frac 1 3$, since the subsequences of edges processed in each iteration do not have fixed positions, and instead depend on when a collision happened in the previous iteration. 
We begin with a technical lemma, which is proved in Appendix~\ref{appendix:proofs}. 

\begin{lemma} \label{lemma:minimas-random-permutation}
    Consider a random permutation of $1, 2, \ldots, N$. 
    Let $M$ be the number of \emph{local minima} in this permutation, i.e. the number of elements $S_i$ smaller than both their neighbours $S_{i-1}$ and $S_{i+1}$ (supposing that the end elements $S_1$ and $S_N$ are also adjacent). Then, it holds that 
    $\texttt{\textup{Pr}}[M \leq \frac {N - 3} {18}] \leq \frac {24} {N - 3}$. 
\end{lemma}

The next lemma will be the core of the lower bound proof. Similarly to Theorem~\ref{theorem:lower-bound-pt}, it analyses the collision probability for a cycle graph. Intuitively, it shows that for a prefix of edges of size $\Omega(\sqrt {N \cdot \log N})$ the probability of collision is almost $1$. The (quite complex) proof is also provided in Appendix~\ref{appendix:proofs}.

\begin{lemma} \label{lemma:cycle-graph-prefix-conflict}
    Consider a sequence of randomly shuffled cycle graph edges of size $N$ and its prefix of size $W \geq C \cdot \sqrt {N \cdot \log N}$, where $C$ is some constant which will be made clear in the proof. Suppose there are $M=pN$ pairs of adjacent colliding edges in this graph, $0<p\leq\frac 1 3$ is a constant. Then the probability that there is no collision in this prefix is $\O\left(\frac 1 {N^2}\right)$.
\end{lemma}

We can now finally state and prove our main lower bound result: 

\begin{theorem}
The number of iterations of any parallel algorithm processing conflict-free prefix of edges (in particular, of the algorithm in Listing~\ref{listing:window-union-find}) is $\Omega \left(\frac {\sqrt {|E|}} {\sqrt {\log {|E|}}}\right)$ with high probability for a cycle graph.
\end{theorem}
\begin{proof}
Again consider a cycle graph.
Initially, when the Union-Find data structure is empty, all connected components are of size and rank $1$. In this case, we assumed that linking by rank or by size use random priorities to break ties. As a result, the number of vertices adjacent to two colliding edges is exactly the number of local minima in the random priorities permutation. Consequently, by Lemma~\ref{lemma:minimas-random-permutation}, the number of pairs of colliding edges is at least $p|E|$ for some constant $p>0$ with probability of failure $\O\left(\frac 1 {|E|}\right)$.

We will prove that $\Omega\left(\frac {\sqrt{|E|}} { \sqrt{\log |E|}}\right)$ iterations are needed to process the first $\frac p 3 |E|$ edges. When an edge is processed, at most $2$ pairs of colliding edges (adjacent to the edge) can stop colliding, since their adjacent connected components change. Thus, after $\frac p 3 |E|$ edges there will still be at least $p|E| - 2 \frac p 3 |E| = \frac {p} {3} |E|$ pairs of colliding edges (which are still adjacent to connected components of size $1$). 

Let $C$ be the constant from Lemma~\ref{lemma:cycle-graph-prefix-conflict} for the case when there are at least $\frac p 3 |E|$ pairs of colliding edges. Let $A_i$ be the event that in the $i$-th iteration, the processed prefix is of size at least $W = C \cdot \sqrt {|E|} \sqrt{\log {|E|}} + 1$. Let $I = \frac {\frac p 3 |E|} {W}$.
Observe that to process $\frac p 3 |E|$ edges in less than $I = \Omega\left(\frac {\sqrt {|E|}} {\sqrt {\log |E|}}\right)$ iterations, at least one of the events $A_i$ for these iterations should happen.

Now if we show that $Pr[\bigcup\limits_{i\leq I} A_i] = \O\left(\frac 1 {|E|}\right)$, the probability of using less than $I$ iterations will be $\O\left(\frac 1 {|E|}\right)$.
For this we required about the existence of enough pairs of colliding edges, which is not satisfied with probability $\O\left(\frac 1 {|E|}\right)$. Summing up these observations, with probability at least $1 - \O\left(\frac 1 {|E|}\right)$ at least $\Omega\left(\frac {\sqrt {|E|}} {\sqrt {\log {|E|}}}\right)$ algorithm iterations are needed, completing the proof. 

We first remark that 
\begin{gather*}
Pr\left[\bigcup\limits_{i\leq I} A_i\right] 
% = Pr[A_1] + Pr[A_2 ~|~ \overline{A}_1] \cdot Pr[\overline{A}_1] + \ldots + \\ + Pr[A_{I} ~|~ \overline{A}_1, \cdots, \overline{A}_{I-1}] \cdot Pr[\overline{A}_1, \cdots, \overline{A}_{I-1}] \leq \\ 
\leq \sum\limits_{i=1}^{I} Pr[A_i ~|~ \overline{A}_1, \ldots, \overline{A}_{i-1}]
\end{gather*}

Consider the probability of such an event $Pr[A_i ~|~ \overline{A}_1, \ldots, \overline{A}_{i-1}]$. Since the previous $A_j$ ($j < i$) are assumed to not occur, less than $\frac p 3 |E|$ edges were processed before. The first edge in the next processing window may have had a collision with some edge in the previous iteration. However, no information is known about edges after it or their order, so we can still assume that the order is uniformly random. 
This lets us to apply Lemma~\ref{lemma:cycle-graph-prefix-conflict} for the next $W-1$ edges after the first one to obtain that the probability of no collision among these edges is $\O\left(\frac 1 {|E|^2}\right)$, i.e., with probability at least $1-\O\left(\frac 1 {|E|^2}\right)$ no more than $W$ edges can be processed in i-th iteration due to data dependencies. Since it is true for any fixed possible start position of the processing window, it is also true in general, i.e., $Pr[A_i ~|~ \overline{A}_1, \ldots, \overline{A}_{i-1}] = \O\left(\frac 1 {|E|^2}\right)$. As a result, $Pr[\bigcup\limits_{i\leq I} A_i] = \O\left(\frac {I} {|E|^2}\right)$, which completes the proof once we substitute the value of $I$.
\end{proof}

In other words, this theorem shows that, even when the algorithm in Listing~\ref{listing:window-union-find} tries to process the whole edge sequence in each iteration, the number of iterations in the worst case is still asymptotically close to the work-efficient version from Theorem~\ref{theorem:main-parallel-algorithm}.

\section{Discussion}

\paragraph{Overview} We have made new progress on understanding the parallel Union-Find problem, providing a new analysis suggesting that classic algorithms for this problem should parallelize well, since data races have low probability, which leads to the first algorithm that is both provably work-efficient and internally-deterministic. 

\paragraph{Practical Performance} 
One interesting question is whether our algorithm is \emph{practical}, i.e. whether ensuring internal determinism imposes significant performance bottlenecks. 
We investigate this question in detail in Appendix~\ref{appendix:experiments}. Figure~\ref{figure:plot-throughput-main} briefly summarizes the results, showing that our algorithm can be competitive with that of Blelloch et al., which does not have theoretical guarantees, and with the fully-concurrent algorithm of Jayanti and Tarjan, which is not internally-deterministic.    

\begin{figure}[t!]
\centering
\includegraphics[width=0.97\linewidth]{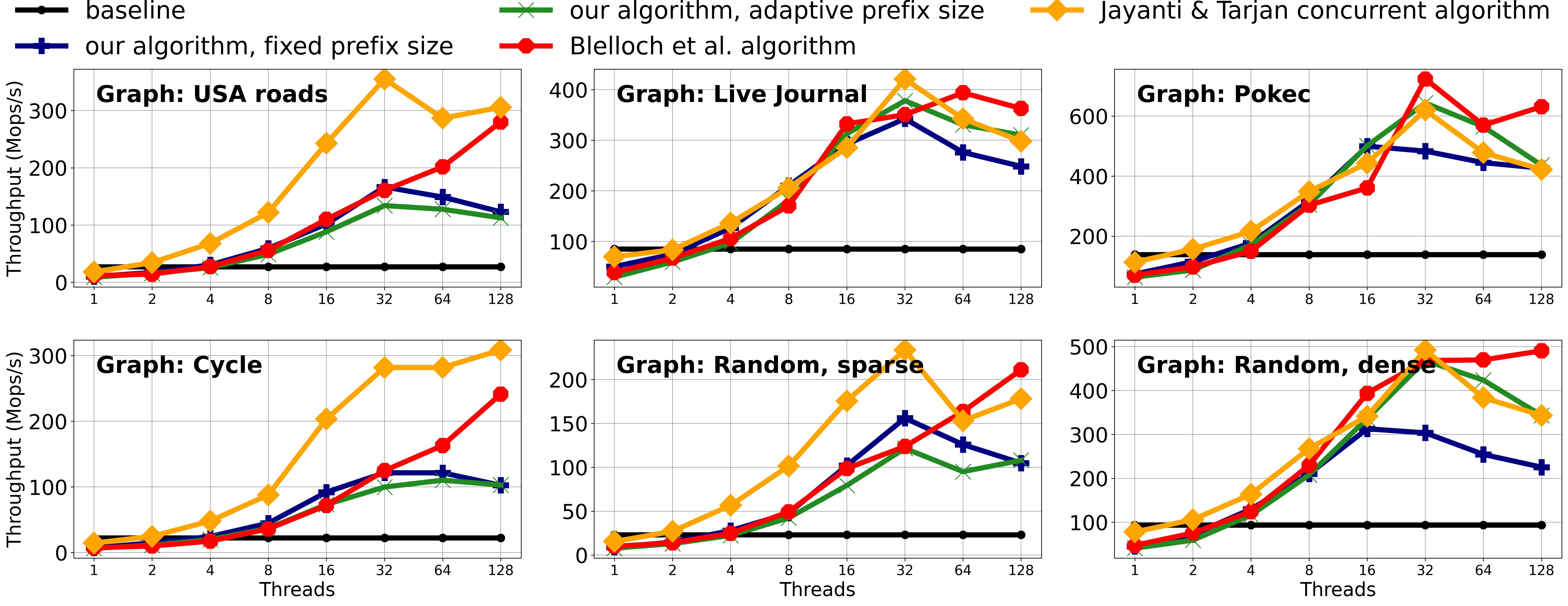}
%\vspace{-0.5em}
\caption{Performance comparison of  parallel and concurrent Union-Find algorithms, relative to an optimized sequential baseline. 
Our algorithm is competitive with prior proposals with weaker guarantees; all algorithms experience a performance drop at $64$ threads, when computation is performed across two NUMA sockets. Please see Appendix for full results.}
\label{figure:plot-throughput-main}
\end{figure}

 \paragraph{Future Work}  
 Our collision analysis can also be employed as a basis for analyzing the parallelization potential of Union-Find in  other computation models, such as execution via relaxed scheduling~\cite{alistarh2018relaxed}. Further, since our potential argument works over arbitrary graph families, our results could extend to the parallelization of other graph problems, such as minimum spanning tree (MST).

\bibliographystyle{unsrt} 

\balance
\bibliography{source.bib}

\appendix

\section{Deferred Proofs}
\label{appendix:proofs}

\begin{replemma}{lemma:minimas-random-permutation}
    Consider a random permutation of $1, 2, \ldots, N$. 
    Let $M$ be the number of \emph{local minima} in this permutation, i.e. the number of elements $S_i$ smaller than both their neighbours $S_{i-1}$ and $S_{i+1}$ (supposing that the end elements $S_1$ and $S_N$ are also adjacent). Then, it holds that 
    $\texttt{\textup{Pr}}[M \leq \frac {N - 3} {18}] \leq \frac {24} {N - 3}$. 
\end{replemma}

\begin{proof}
    Let $X_i$ be indicator variables ``$i$ is the position of a local minimum''. Then $\E [M] = \E [\sum\limits_i X_i]$.  

    Trivially, $\E [X_i] = \frac 1 3$, since among all permutations of element $i$ and its neighbours exactly $2$ out of $3!=6$ result in a local minima at position $i$. Moreover, $Var[X_i] = \E [X_i^2] - \E [X_i]^2 = \E [X_i] - \E [X_i]^2 = \frac 1 3 - \frac 1 9 = \frac 2 9$.

    Let $M^* = X_3 + X_6 + X_9 + \ldots = \sum\limits_{i=1}^{N/3} X_{3i}$, i.e. the sum over every third element of the permutation. Since the elements and their adjacent elements in this sum are all disjoint, the indicator variables are independent. In other words, for $X_{3i}$, even when the values of all other indicator variables in this sum are known, they provide no information about the order of element $3i$ and its neighbours. As a result, $Var[M^*] = \sum\limits_i Var[X_{3i}] = \lfloor \frac {N} {3} \rfloor \cdot \frac 2 9$.

    We also know that $\E [M^*] = \sum\limits_i \E[X_{3i}] = \lfloor \frac {N} {3} \rfloor \cdot \frac 1 3$.

    By Chebyshev's inequality, we obtain that 
    $\texttt{Pr}[\lfloor \frac {N} {3} \rfloor \cdot \frac 1 3 - M^* \geq \lfloor \frac {N} {3} \rfloor \cdot \frac 1 6] \leq \frac {\lfloor \frac {N} {3} \rfloor \cdot \frac 2 9} {\left(\lfloor \frac {N} {3} \rfloor \cdot \frac 1 6\right)^2} \leq \frac {24} {N-3}$.

    Since $M^* \leq M$, this means that $\texttt{Pr}[M \leq \frac {N - 3} {18}] \leq \frac {24} {N-3}$.
\end{proof}

\begin{lemma} \label{lemma:sum-combinations}
    $\sum\limits_{k=0}^{pN} \binom {N} {k} \leq 2^{H(p)N - \ln N / 2 + \O(1)}$ for a constant $0 < p <\frac 1 2$, where $H(p)$ is the entropy function, which is has a constant value here as well.
\end{lemma}

\begin{replemma}{lemma:cycle-graph-prefix-conflict}
    Consider a sequence of randomly shuffled cycle graph edges of size $N$ and its prefix of size $W \geq C \cdot \sqrt {N \cdot \log N}$, where $C$ is some constant which will be made clear in the proof. Suppose there are $M=pN$ pairs of adjacent colliding edges in this graph, $0<p\leq\frac 1 3$ is a constant. Then the probability that there is no collision (i.e., no pair of colliding edges) in this prefix is $\O\left(\frac 1 {N^2}\right)$.
\end{replemma}
\begin{proof}

Since the sequence is randomly shuffled, we can think of the prefix as if edges are (uniformly) randomly added to it one by one without replacement. Consider a recurrence for the probability of ``collision'' for this random process. Let $T(t, m, k)$ be the probability of a future collision at the moment when $t$ edges are already added to the prefix, $m$ pairs of colliding edges are not in the current prefix and $k$ pairs of colliding edges have exactly one edge in it. As a result, these $k$ remaining edges will cause a collision when added to the prefix, $2m$ edges will touch one edge in a colliding pair, and the other $N-t-2m-k$ edges will not change $m$ and $k$. More formally:

 $$T(t, m, k) = \frac {2m} {N-t} \cdot T(t+1, m-1, k+1) + \frac {N-t-2m-k} {N-t} T(t+1, m, k).$$

The base of this recurrence is $T(W, m, k)=1$ for any $m, k$.
Note that, in this recurrence, we can replace $k$ with $M-m$ in order to calculate $T(0, M, 0)$, since $m$ and $k$ are changing synchronously. This leads to:

 $$T(t, m) = \frac {2m} {N-t} \cdot T(t+1, m-1) + \frac {N-t-m-M} {N-t} T(t+1, m).$$

To calculate $T(0, M)$ (the probability of no collision in the prefix), we use this recurrence relation repeatedly until the moment when $t = W$ (the recurrence base case). This way $T(0, M)$ is expanded to a sum of $2^W$ terms, each of which is multiplied by $T(W, m) = 1$ for some $m$. We group these terms by the final value of $m$.

Next, we claim that all terms in the same group have the same value 
$V_m = \frac {(2M) \cdot (2M-2) \cdots (2m+2) \cdot (N-2M)\cdot(N-2M-1) \cdots (N-M-m-W+1)} {N\cdot (N-1) \cdots (N-W+1)}$. This follows from the observation that all terms in the same group had the same number of ``steps'' of type $(t,m) \rightarrow (t+1, m-1)$ (which is $M-m$) and the order of these steps among $W$ total steps does not influence the final multiplier.

Consequently, $$T(0,M) = \sum\limits_{m=M-W}^{M} \binom {W} {M-m} V_m = \sum\limits_{i=0}^{W} \binom {W} {i} V_{M-i}$$

Now we will prove that this sum is $\O\left(\frac 1 {N^2}\right)$. To do this, we first show that a small prefix of this sum is $\O\left(\frac 1 {N^2}\right)$ due to small binomial coefficients and then show that the rest of the sum is $\O\left(\frac 1 {N^2}\right)$ due to arising collision probabilities. 

Consider $\sum\limits_{i=0}^{lW} \binom {W} {i} V_{M-i}$ for some constant $0 < l < \frac 1 2$, which will be chosen later. For $i \leq lW$, 

\begin{gather*}
V_{M-i} \leq \frac {(N-2M)(N-2M-2) \cdots (N-2M-2(1-l)W+2)} {N(N-1) \cdots (N-(1-l)W+1)}  \leq \frac {(N-2M)^{(1-l)W}} {N^{(1-l)W}} \leq (1-2p)^{(1-l)W}
\end{gather*}

By Lemma~\ref{lemma:sum-combinations}, $\sum\limits_{i=0}^{lW} \binom {W} {i} \leq 2^{H(l)W - \ln W / 2 + \O(1)}$.

By combining these inequalities, we get that $\sum\limits_{i=0}^{lW} \binom {W} {i} V_{M-i} \leq (1-2p)^{(1-l)W} \cdot 2^{H(l)W - \ln W / 2 + \O(1)} = 2^{(H(l) - (1-l)\log \frac 1 {1 - 2p})W - \ln W / 2 + \O(1)}$. As a result, we can choose constant $l$ so that $H(l) - (1-l)\log \frac 1 {1 - 2p} < 0$, which would mean that this sum decreases exponentially with $W$ and in particular is $\O\left(\frac 1 {N^2}\right)$.

What is left is to prove that $\sum\limits_{i=lW+1}^{W} \binom {W} {i} V_{M-i}$ is also small.

First consider a similar recurrence relation $\tilde{T}(t, m) = \frac {{\color{red}{M+m}}} {N-t} \tilde{T}(t+1,m-1) + \frac {N-t-m-M} {N-t} \tilde{T}(t+1, m)$. Intuitively, the difference in these recurrence relations is caused exactly by the probability of collision. It is easy to see that $\tilde{T}(t, m) = 1$ solves the recurrence. Similarly, we can represent $\tilde{T}(0, M)$ as $\sum\limits_{i=0}^{W} \binom {W} {i} \tilde{V}_{M-i}$, which also equals $1$.

It can be shown that $\frac {V_{M-i}} {\tilde{V}_{M-i}} = \frac {(2M) (2M-2) \cdots (2M-2i+2)} {(2M) (2M-1) \cdots (2M-i+1)}$.
If we remove the first three quarters of these multipliers (since the corresponding fractions are $\leq 1$), we obtain that \begin{align*} \frac {V_{M-i}} {\tilde{V}_{M-i}} &\leq \frac {(2M - \frac 3 2 i) \cdots (2M - 2i+2)} {(2M - \frac 3 4 i) \cdots (2M - i + 1)} \leq \left(\frac {2M - \frac 3 2 i} {2M - i + 1} \right)^{i/4} \leq \left(1 - \frac {i} {4M}\right)^{i/4} \leq e^{-i^2/16M} \end{align*}.

The last step follows from Bernoulli's inequality. When $i \geq lW$, $\frac {V_{M-i}} {\tilde{V}_{M-i}} \leq e^{-(lW)^2/16M} = e^{-l^2C^2 N \log N / 16pN} = e^{-l^2C^2 \log N / 16 p}$. Now, given that $l$ and $p$ are constants, we can fix constant $C$ so that $-l^2C^2\cdot \ln 2/16p < -2$. This would mean that for $i \geq lW$, $\frac {V_{M-i}} {\tilde{V}_{M-i}} \leq \frac 1 {N^2}$.

\begin{gather*}
\sum\limits_{i=lW+1}^{W} \binom {W} {i} V_{M-i} = \sum\limits_{i=lW+1}^{W} \binom {W} {i} \tilde{V}_{M-i} \cdot \frac {V_{M-i}} {\tilde{V}_{M-i}} \leq\\ \leq \sum\limits_{i=lW+1}^{W} \binom {W} {i} \tilde{V}_{M-i} \cdot \frac 1 {N^2} \leq \frac 1 {N^2} \sum\limits_{i=0}^{W} \binom {W} {i} \tilde{V}_{M-i} = \frac 1 {N^2}
\end{gather*}

Summing up, the probability that there is no collision in the prefix $T(0, M) = \O\left(\frac 1 {N^2}\right)$.

\end{proof}

\pagebreak
\section{Full Experimental Analysis}
\label{appendix:experiments}
\paragraph{Scenario} We compare our parallel algorithms with other algorithms on a connected components benchmark. Basically, an algorithm receives a sequence of edges of some graph as an input and then inserts them in the data structure in parallel. Concurrent algorithms also receive a sequence of edges but then use light-weight load balancing to distribute edges evenly and to avoid the situation when some threads have many edges left to insert while other threads already completed their part. Specifically, concurrent algorithms split the sequence into batches and employ \texttt{Fetch-And-Increment} when a thread needs to pick the next batch.

\paragraph{Algorithms and Graph Inputs} We benchmark the following algorithms: the proposed internally-deterministic algorithm with fixed prefix size, as well as an optimized variant with adaptive prefix size, the efficient heuristic algorithm of Blelloch et al.~\cite{ParallelConnectedComponents}, as well as the fully-concurrent algorithm of Jayanti and Tarjan~\cite{JayantiT21} with the optimizations from~\cite{ConnectIt,AlistarhFK19}.

We execute on both real-world and synthetic graphs, with tens of millions edges. The USA roads graph and Kron graph were taken from DIMACS Challenge competitions~\cite{dimacs:challenge9, dimacs:challenge10}. The source of other real graphs is SNAP Graph Collection~\cite{snapnets}. Random graphs of various densities were generated using Erdos-R{\'e}nyi model~\cite{ErdosRenyi}.

For fairness reasons, we use the same path compaction and linking techniques in all evaluated algorithms. Specifically, we choose one-try splitting and linking by random priority techniques as they are shown to be among the fastest in practice and the simplest in implementation~\cite{AlistarhFK19,ConnectIt}. It should be noted that \emph{concurrent} linking by random priority was proved to be theoretically efficient by Jayanti and Tarjan~\cite{JayantiT21} but only under the assumption that linearization order of unites is independent from random priorities, which may be not satisfied in practice. 

\begin{figure}[b]
\centering
\includegraphics[width=\linewidth]{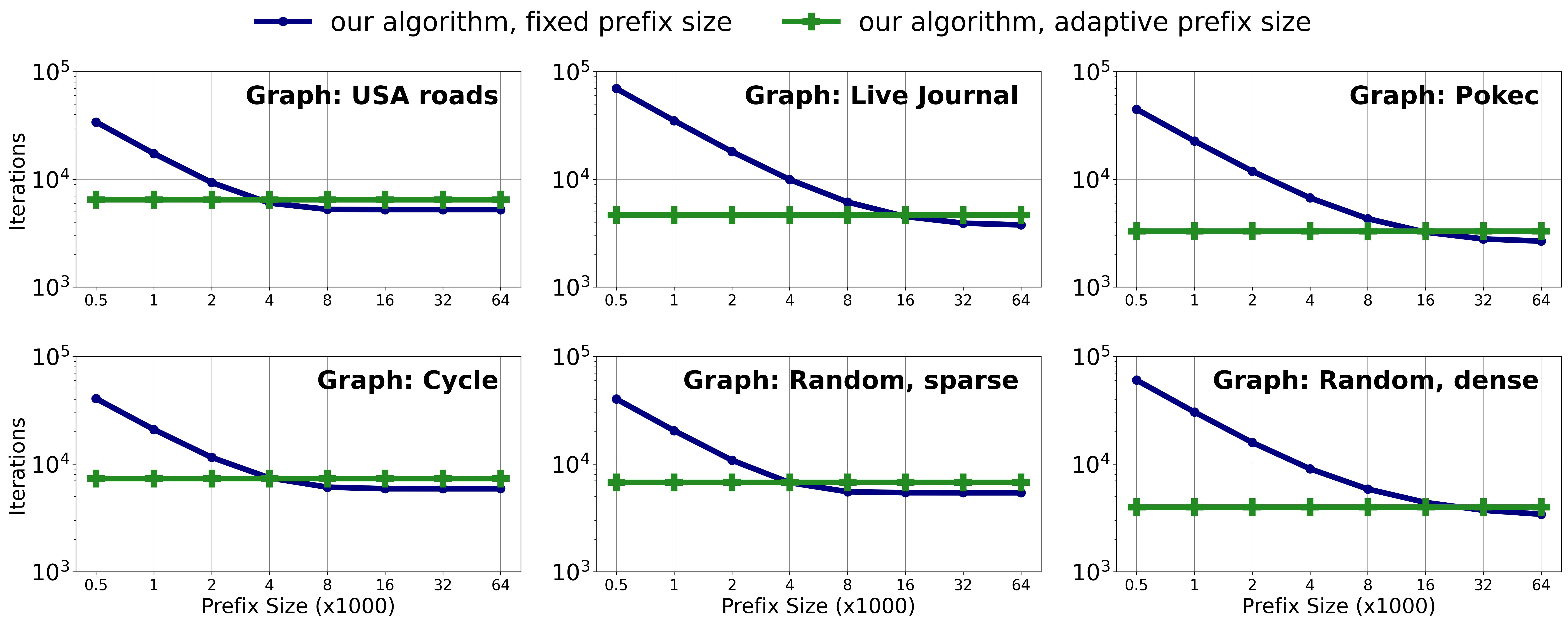}
%\vspace{-0.5em}
\caption{Dependency of the prefix size on the number of iterations in our parallel algorithm. The adaptive algorithm chooses the prefix size automatically and still achieves near-optimal number of iterations.}
\label{figure:plot-iterations}
\end{figure}

\begin{figure}[htb!]
\centering
\includegraphics[width=\linewidth]{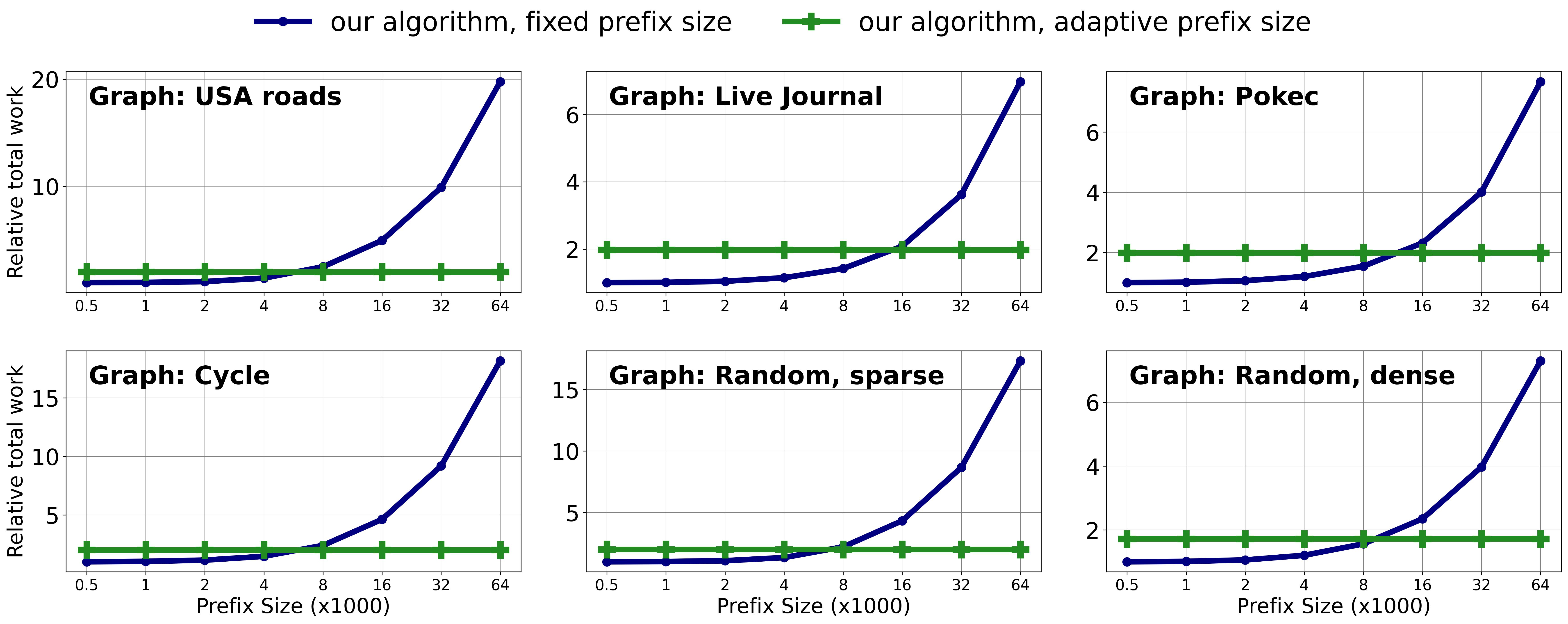}
%\vspace{-0.5em}
\caption{Comparison of the total work in the adaptive parallel algorithm and non-adaptive algorithm for various prefix sizes relative to the amount of work in the sequential algorithm.}
\label{figure:plot-work}
\end{figure}

\paragraph{Hardware and Software}
The benchmarks are implemented in Java (OpenJDK 17.0.4) and are run with Java Microbenchmark Harness (JMH) to mitigate JIT and garbage collector effects on the performance~\cite{jmh}. We use an Intel Xeon Gold 6150 machine with 4 NUMA sockets, each with 18 cores and hyper-threading enabled, and with 512GB RAM.

\paragraph{Parameters and Optimizations} 
In our parallel Union-Find algorithm, the prefix size $S$ provides a trade-off between scalability and ``useless'' work. The larger prefix size is, the more threads can process it in parallel, but the more work will potentially be wasted because of collisions. Corollary~\ref{corollary:window-work-efficient} suggests that the total work asymptotically stays the same if $S=\O(|E|^{\frac 1 3-\varepsilon})$, so it is reasonable to choose the prefix size close to $|E|^{\frac 1 3}$. However, our experiments show that there is a simpler and more general adaptive approach.
The adaptive parallel algorithm changes the prefix size $S$ between iterations. Specifically, when no collisions occur in an iteration, the algorithm doubles $S$. When there is a collision in the first half of the prefix, for the next iteration the prefix size will be halved. Figures~\ref{figure:plot-iterations} and \ref{figure:plot-work} show that this approach yields near-optimal number of iterations, while the amount of wasted work keeps being limited. Surprisingly, for some points the adaptive algorithm is better both in terms of the number of iterations and the total work, which may mean that for some graphs it is better to change the prefix size $S$ during the execution, rather than to keep it fixed.

One of the main properties of the presented parallel algorithm is that it is internally deterministic. Yet, it is unclear whether to get this property we sacrifice Union-Find performance. To answer this question, we compare the algorithm to the concurrent Union-Find algorithm of Jayanti et al.~\cite{Tarjan2014}, in particular to its fast implementation from~\cite{ConnectIt, AlistarhFK19}, and to Blelloch et al. parallel algorithm~\cite{InternallyDeterministicParallel}. We emphasize that this is not a fair comparison, since neither of these algorithms are internally deterministic. When implementing our algorithm, we use the same optimizations as in~\cite{WorkEfficientUnionFind}. Specifically, we do not implement the BFS-like algorithm for path compaction and instead compact paths during searches using \texttt{Compare-And-Set}, as in the concurrent algorithm.

\paragraph{Performance Comparison}
The performance benchmark results are available in Figure~\ref{figure:plot-throughput}. First of all, we can see that our adaptive parallel algorithm matches or outperforms the non-adaptive algorithm with fixed prefix size $S$. As expected, the heuristic algorithm of Blelloch et al. has better performance than our algorithm, since our algorithm has additional restriction of being internally deterministic but the performance differences are usually small. The well-optimized fully-concurrent approach of Jayanti and Tarjan is superior to parallel algorithms on sparse graphs (USA roads, the cycle graph, and the random sparse graph). The reason is that, for sparse graphs, there is almost no contention between different threads, and thus, the concurrent algorithm has almost no synchronization cost. In dense graphs, the asynchronous nature of the concurrent algorithm does not provide any benefits over the parallel algorithms and incurs additional synchronization work.

\begin{figure}[htb!]
\centering
\includegraphics[width=\linewidth]{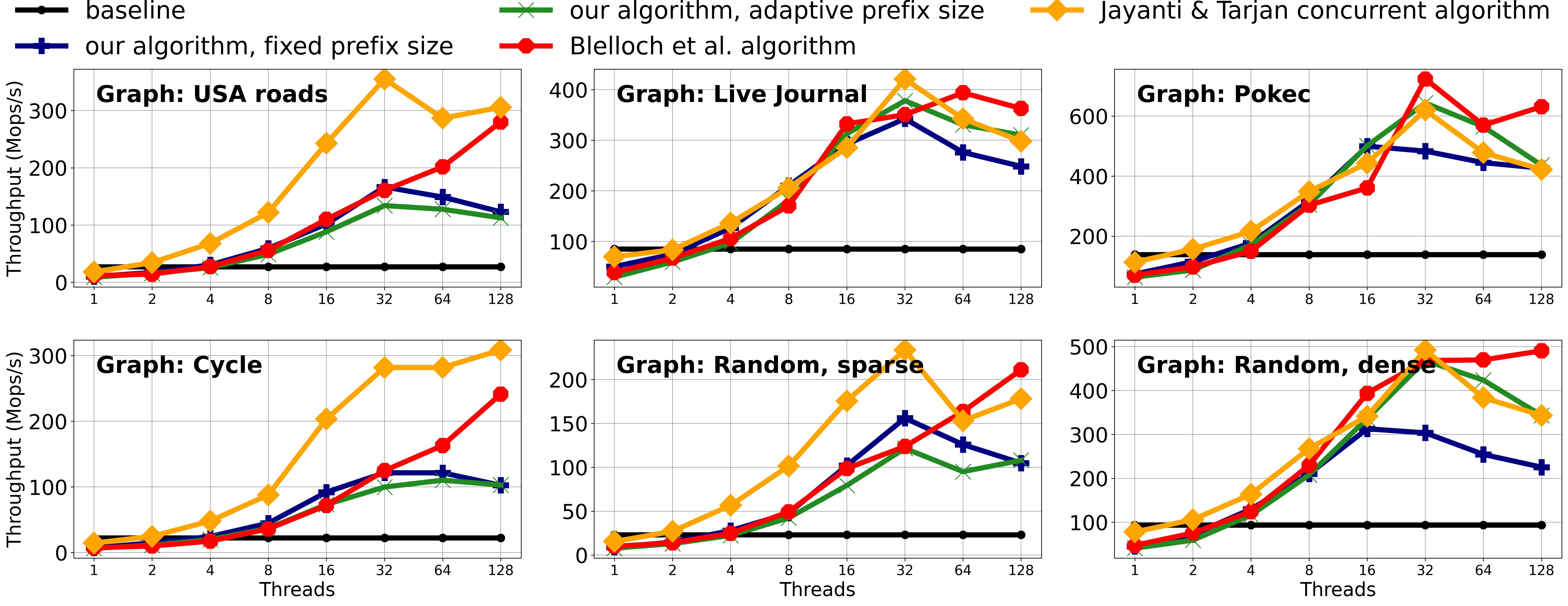}
%\vspace{-0.5em}
\caption{Performance comparison of the parallel and concurrent Union-Find algorithm. The baseline is the sequential algorithm. All algorithms experience a performance drop for more than $64$ threads, because $\leq 32$ threads can fit in a single NUMA socket on the benchmark machine. The prefix size $S$ for our parallel algorithm and for Blelloch at al. algorithm was chosen optimally based on experiment results.}
\label{figure:plot-throughput}
\end{figure}

\end{document}